\documentclass{article}
\usepackage{lipsum}
\usepackage{xifthen}
\usepackage{listings}
\RequirePackage{amsthm,amsmath,amssymb}
\usepackage{natbib}
\usepackage{authblk}
\usepackage{float}
\usepackage{rotating}

\newcounter{num}
\makeatletter
\def\@maketitle{%
  \newpage
  \null
  \vskip 2em%
  \begin{center}%
  \let \footnote \thanks
    {\Large\bfseries \@title \par}%
    \vskip 1.5em%
    {\normalsize
      \lineskip .5em%
      \begin{tabular}[t]{c}%
        \@author
      \end{tabular}\par}%
    \vskip 1em%
    {\normalsize \@date}%
  \end{center}%
  \par
  \vskip 1.5em}
\makeatother

\newcommand{\doBlank}[1]{}

\def\?#1{}

\newcommand{\dif}{\mathrm{d}}

\numberwithin{equation}{section}
\theoremstyle{plain}

\newtheorem{assumption}{Assumption}
\newtheorem{lemma}{Lemma}[section]
\newtheorem{definition}{Definition}[section]
\newtheorem{proposition}{Proposition}[section]
\newtheorem{corollary}{Corollary}[section]

\newtheorem{theorem}{Theorem}

\newcommand{\writetitle}{0}
\newcommand{\mytitle}[1]
{   \ifthenelse{\writetitle=1}{}{}
}

\newread\mysource
\begin{document}
\title{Ergodicity of Markov chain Monte Carlo with reversible proposal}
\author{Kengo KAMATANI%
\thanks{Supported in part by Grant-in-Aid for Young Scientists (B) 24740062 and CREST JST.}}
\affil{Graduate School of Engineering Science and Center for Mathematical Modeling and Data Science, Osaka University}

\date{Dated: \today}

\maketitle
\begin{abstract} 
We describe  ergodic properties of some Metropolis-Hastings (MH) algorithms for heavy-tailed target distributions. The analysis usually falls into  sub-geometric ergodicity framework but we prove that the mixed preconditioned Crank-Nicolson (MpCN) algorithm has geometric ergodicity even for heavy-tailed target distributions. 
This useful property  comes from the fact that the MpCN algorithm becomes a random-walk Metropolis algorithm under suitable transformation. 
\end{abstract}

\section{Introduction}

In Bayesian analysis, direct calculation of integral is usually quite difficult  especially  for high-dimension and/or heavy-tailed target distributions. 
Markov chain Monte Carlo (MCMC) methods such as Metropolis-Hastings (MH) algorithm provides a useful recipe for the approximation of the integral.  

%

Ergodic properties for heavy-tailed case were handled mostly by sub-geometric drift condition (see e.g. \citet{MR1285459,MR1890063,MR1947960,FM1})
since most MCMC do not satisfy geometric drift condition. Application of sub-geometric drift condition to MCMC includes \citet{FM2, MR1996270, FM1,MR2396939}, 
and \citet{MR2648752}. 
On the other hand some other MCMC can be geometrically ergodic. This includes independent sampler and position dependent variance MH algorithm on $\mathbb{R}$
such as \citet{1008.5227, 1507.05780}. Note that independent sampler is very sensitive for the choice of the proposal distribution, and position dependent methods have difficulty in high-dimension which may negatively affect ergodic properties. 

In this paper, we consider geometric ergodicity  for multidimensional heavy-tailed and light-tailed target distributions. 
Recently the mixed preconditioned Crank-Nicolson (MpCN) algorithm was considered in \cite{arXiv:1412.6231}. The method  has good high-dimensional properties even for heavy-tailed target distributions.  As the number of dimension $d\rightarrow\infty$, the number of iteration until convergence is $O(d)$ whereas the random-walk Metropolis one is $O(d^2)$. 
To prove ergodicity, we provide  the key property, the \textbf{random-walk Metropolis} property for the MpCN kernel: The MpCN kernel becomes a random-walk Metropolis kernel under suitable transformation. Thus MpCN is considered to be an extreme case of variable transformation methods (see \citet{MR2556773} and \citet{MR3097969}). 
By using this fact, it is rather straightforward to show geometric ergodicity for fairly general class of target distributions in $\mathbb{R}^d$. 

The main result is summarized in the next theorem. The formal definition of the MpCN kernel is in Section \ref{sec:mhkernels} and the proof is deferred to Sections \ref{sec:regular} and \ref{sec:rapid}.  

\begin{theorem}
The MpCN kernel is geometrically ergodic for the target probability distribution $\Pi(\dif x)=\pi(x)\dif x$ on $\mathbb{R}^d$ such that 
\begin{description}
\item[Heavy-tailed class]
 $\pi(x)$ is strictly positive continuous function such that 
\begin{align*}
	\lim_{r\rightarrow\infty}\frac{\pi(rx)}{\pi(r1)}=\|x\|^{-\alpha}
\end{align*}
for some $\alpha>d$ where the above convergence is locally uniform in $x$. 
\item[Light-tailed class]
$\pi(x)$ is strictly positive differentiable  function such that 
\begin{align*}
	\lim_{r\rightarrow\infty}\frac{\pi(rsx)}{\pi(rx)}=
	\left\{\begin{array}{cc}
	0&\mathrm{if}\ 1<s\\
	+\infty&\mathrm{if}\ 1>s
	\end{array}
\right.
\end{align*}
for any $x\neq 0$, and satisfies a curvature condition 
	\begin{align*}
		\limsup_{x\rightarrow\infty}\left\langle\frac{x}{\|x\|},\frac{\nabla\log\pi(x)}{\|\nabla\log\pi(x)\|}\right\rangle <0
	\end{align*}
	\end{description}
\end{theorem}

The heavy-tail class includes (a) polynomial target densities considered in \citet{MR2396939} (Section 3.3), and 
the light-tailed class includes (b) super-exponential densities in \citet{MR1731030} (Section 4) and (c) exponential densities in \citet{FM2} (Assumption D). Note that the random-walk Metropolis algorithm is geometrically ergodic only for (b) (Theorem 4.3 of \citet{MR1731030}).

The rest of the paper is organized as follows. In Section \ref{sec:mpcn} MpCN algorithm is introduced as a MH kernel with reversible proposal. In this section, the  random-walk Metropolis property is defined and proved that the MpCN kernel has the property. Section \ref{sec:ergodicity} provides ergodic properties of MpCN kernel. 

We finish the section with notation that will be used through the paper. $N_d(\mu,\Sigma)$ is the $d$-dimensional normal distribution with mean $\mu$ and variance covariance matrix $\Sigma$, and $\phi_d(x)$ is the density of $N_d(0,I_d)$ where $I_d$ is the $d\times d$ identity matrix. 
$\mathcal{L}(X)$ is the law of the random variable $X$. 

\section{The MpCN kernel}\label{sec:mpcn}

In this section we describe the mixed preconditioned Crank-Nicolson (MpCN) algorithm
as an MH kernel with reversible proposal kernel. For general background on 
Markov chain we refer to \citet{N} and \citet{MR2509253} and 
 MCMC to \citet{TierneyAOS94}, and \citet{handbookmcmc}.

\subsection{Metropolis-Hastings kernels with reversible proposals}\label{sec:mhkernels}

Let $(E,\mathcal{E})$ be a measurable space and let $P$ be a (probability) transition kernel and $\Pi(\dif x)$ be a probability measure. 
The transition kernel $P$ is called $\nu$-\textbf{reversible}
if $\nu(\dif x)P(x,\dif y)=\nu(\dif y)P(y,\dif x)$ for a $\sigma$-finite measure $\nu$. 
Let $\tilde{\Pi}$ be a $\sigma$-finite measure on $(E,\mathcal{E})$, 
and $\pi$ and $\tilde{\pi}$ be the densities of $\Pi$ and $\tilde{\Pi}$ with respect to a $\sigma$-finite measure. If transition kernel $Q$ is $\tilde{\Pi}$-reversible, 
\textbf{Metropolis-Hastings (MH) kernel} $P$ (with reversible proposal) is defined by 
\begin{align*}
	P(x,\dif y)=Q(x,\dif y)\alpha(x,y)+\delta_x(\dif y)\left(1-\int_{z\in E}Q(x,\dif z)\alpha(x,z)\right)
\end{align*}
where 
\begin{align*}
	\alpha(x,y)=\min\left\{1,\frac{\pi(y)\tilde{\pi}(x)}{\pi(x)\tilde{\pi}(y)}\right\}.
\end{align*}
We call $Q$ the proposal kernel of $P$. MH kernel is $\Pi$-reversible. 

In this paper, three MH kernels on Euclidean space will be studied.  
Assume $d\ge 2$. Let $S^{d-1}$ be the unit sphere in $\mathbb{R}^d$ defined by $\|x\|=1$ where $\langle\cdot,\cdot\rangle$ is the Euclidean inner product and 
 $\|x\|=\langle x,x\rangle^{1/2}$. 
A probability measure $\Gamma$ on $\mathbb{R}^d$ is called symmetric about the origin if $\Gamma(A)=\Gamma(-A)$ for any Borel set $A$ where $-A=\{-x;x\in A\}$.

\begin{definition}[RWM kernel]
The random-walk Metropolis kernel uses $Q(x,\dif x^*)=\Gamma(\dif x^*-x)$ where 
the probability distribution $\Gamma$ is  symmetric about the origin.
In this case $\tilde{\Pi}$ is the Lebesgue measure. 
Its ergodic properties were studied in \citet{MT2, RT, MR1731030}
and \citet{FM2}. 
\end{definition}

\begin{definition}[pCN kernel]
The preconditioned Crank Nicolson (pCN) kernel (\citet{000259234400002n.d.}) uses 
\begin{align*}
	x^*\leftarrow \rho^{1/2} x+(1-\rho)^{1/2}w,\ w\sim N_d(0,I_d)
\end{align*}
as the proposal kernel. In this case $\tilde{\Pi}$ is the standard normal distribution. This method has mainly been studied by high-dimensional analysis, see e.g., \citet{MR3262508}.
%
\end{definition}

To obtain a better mixing property, we consider scale mixture version of the pCN kernel. See \citet{arXiv:1412.6231} for more background and high-dimensional asymptotic theory. 

\begin{definition}[MpCN kernel] 
The mixed pCN (MpCN) kernel  (\citet{arXiv:1412.6231}) uses
\begin{align*}
&r\sim \mathrm{Gamma}(d/2,\|x\|^2/2),\ w\sim N_d(0,I_d), \mathrm{and}\\
&x^*\leftarrow \rho^{1/2} x+(1-\rho)^{1/2} r^{-1/2}w. 
\end{align*}
as the proposal kernel. 
In this case $\tilde{\Pi}(\dif x)=\|x\|^{-d}\dif x$. 
\end{definition}

In the above, $\mathrm{Gamma}(\nu,\alpha)$ is the Gamma distribution with the shape parameter $\nu$
and the scale parameter $\alpha$ with the probability distribution function $\propto x^{\nu-1}\exp(-\alpha x)$. 
We usually set $\rho=0.8$. 
Obviously, the proposal kernels for RWM and pCN are reversible
and it is also true for MpCN kernel. See Lemma 2.1 of \citet{arXiv:1412.6231} for the proof. 

Since $\|\tilde{w}\|^2$ follows the chi-squared distribution $\mathrm{Gamma}(d/2,1/2)$, we have another useful expression
\begin{align}\label{eq:xstar}
	x^*\leftarrow \rho^{1/2}x+(1-\rho)^{1/2}\|x\|\frac{w}{\|\tilde{w}\|}
\end{align}
for the proposal of the MpCN kernel, where $w, \tilde{w}\sim N_d(0,I_d)$ are independent. 
By this notation, 
\begin{align}\label{eq:mpcn_expression}
\|x^*\|^2=\|x\|^2\left(\rho+2\sqrt{\rho(1-\rho)}\frac{\|w\|}{\|\tilde{w}\|}v+(1-\rho)\frac{\|w\|^2}{\|\tilde{w}\|^2}\right)
\end{align}
where $v=\langle w/\|w\|,x/\|x\|\rangle$. The law of $v$ is the first element of the uniform distribution on $S^{d-1}$ and it is independent from $\|w\|$ and $\|\tilde{w}\|$.
Therefore the law of $\|x^*\|/\|x\|$ does not depend on $x$. In particular, the law of 
\begin{align}\label{eq:xi}
	\xi(x)=\log(\|x^*\|^2)-\log(\|x\|^2)=\log\left(\frac{\|x^*\|^2}{\|x\|^2}\right)
	\end{align}
does not depend on $x$. 

\subsection{Random walk property} 
In this section we will present the random-walk Metropolis (RWM) property of the MpCN kernel which is the key for the proof of ergodicity in Section \ref{sec:ergodicity}. 
Let $\Psi^{-1}A=\{x; \Psi(x)\in A\}$. 

\begin{definition}
A transition kernel  $P$ on $(E,\mathcal{E})$ has the \textbf{random-walk property}
with respect to $\Psi:E\rightarrow \mathbb{R}^d$ if  there exists a probability distribution $\Gamma$ which is symmetric about the origin such that
\begin{align*}
	P(y,\Psi^{-1}A)=\Gamma(A-x)
\end{align*}
for all $x\in\mathbb{R}^d, y\in E$ and Borel set $A$ such that $\Psi(y)=x$. 
The MH kernel $P$ has the \textbf{RWM property} with respect to $\Psi$
if its proposal kernel $Q$ has the random-walk property with respect to $\Psi$. 
\end{definition}

 A few methods with this property have been proposed in the literature, including multiplicative random walk in \cite{MR2001384} and transformation method in 
\citet{MR3097969}. 

\begin{proposition}\label{prop:rwp}
The law of $\xi(x)$ in (\ref{eq:xi}) is symmetric about the origin and does not depend on $x$. 
In particular, the MpCN kernel has the RWM property with respect to $\Psi(x)=\log(\|x\|^2)$. 
\end{proposition}

\begin{proof}
By expression (\ref{eq:mpcn_expression}), the law of $\xi(x)$ and  $\xi(\tilde{w})$ are the same as described above. Note that
\begin{align*}
\xi(\tilde{w})=\log\left(\|\rho^{1/2}\tilde{w}+(1-\rho)^{1/2}w\|^2\right)-\log\left(\|\tilde{w}\|^2\right).
\end{align*}
Moreover,  there exists exchangeability  $\mathcal{L}(\tilde{w},\rho^{1/2}\tilde{w}+(1-\rho)^{1/2}w)=\mathcal{L}(\rho^{1/2}\tilde{w}+(1-\rho)^{1/2}w,\tilde{w})$. 
Therefor the law of $\xi(\tilde{w})$ is symmetric about the origin since
\begin{align*}
	\mathbb{P}_x(\xi(\tilde{\omega})>t)&=\mathbb{P}_x(\log\|\rho^{1/2}\tilde{w}+(1-\rho)^{1/2}w\|^2-\log\|\tilde{w}\|^2>t)\\
	&=\mathbb{P}_x(\log\|\tilde{w}\|^2-\log\|\rho^{1/2}\tilde{w}+(1-\rho)^{1/2}w\|^2>t)=\mathbb{P}_x(\xi(\tilde{\omega})<-t)\ (t\in\mathbb{R}). 
\end{align*}
 Thus the claim holds by putting $\Gamma=\mathcal{L}(\xi(\tilde{w}))=\mathcal{L}(\log\|\rho^{1/2}\tilde{w}+(1-\rho)^{1/2}w\|^2-\log\|\tilde{w}\|^2)$. 
\end{proof}

\section{Ergodicity}\label{sec:ergodicity}
We have introduced the MpCN kernel in Section \ref{sec:mpcn} as an extension to the pCN kernel and showed it to have the RWM property. For this reason, the ergodic properties of the MpCN kernel can be derived in the same way as that of the RWM kernel. We consider heavy-tailed target distributions in Section \ref{sec:regular} and light-tailed target distributions in Section \ref{sec:rapid}. We prepare  Section \ref{sec:necessary} for necessary condition for geometric ergodicity.   We will conclude that unlike the RWM and pCN kernels, the MpCN kernel is geometrically ergodic for very wide class of  target distributions. 

I will begin by reviewing a few elementary properties of transition kernels. Our notation and terminologies generally follow those of \cite{MR2509253}. 
Let $P(x,\dif y)$ be a transition kernel on a measurable space $(E,\mathcal{E})$. 
We define
\begin{align*}
	Ph(x)=\int_y P(x,\dif y)h(y),\ (\nu P)(\dif y)=\int_x\nu(\dif x)P(x,\dif y)
\end{align*} 
for any measurable function $h(x)$ and signed measure $\nu$ if the right-hand side exists. 
A probability measure $\Pi$ is called the \textbf{invariant probability measure} if $\Pi P=\Pi$
and $P$ is called $\Pi$-invariant. 
Let $P^0(x,\dif y)=I(x,\dif y):=\delta_x(\dif y)$
and $P^{k+1}(x,\dif y)=\int_zP(x,\dif z)P^k(z,\dif y)\ (k\ge 0)$. The kernel $P$ is called $\Pi$-\textbf{irreducible} if $\Pi$ is absolutely continuous with respect to $\sum_{k=1}^\infty P^k(x,\cdot)$ for any $x\in E$. 
A set $C\in\mathcal{E}$ is called \textbf{small set} if 
	\begin{align}\label{eq:smallset}
	P^k(x,\cdot)\ge \delta\nu\ (x\in C)
\end{align}
for some $k\in\mathbb{N}$, $\delta\in (0,1)$, and a probability measure $\nu$. 
We require usual assumptions throughout in this paper: (a) $P$ is  $\Pi$-irreducible (b) $P$ is $\Pi$-invariant (c) there exists a small set $C\in\mathcal{E}$ such that $\Pi(C)>0$ and (d) $\Pi$ is not singular, that is, $\Pi(\{x\})<1$ for $x\in E$. Note that if $\mathcal{E}$ is countably generated, (c) comes from (a) (Proposition 2.6 of \cite{N}). 

Let $V:E\rightarrow [1,\infty]$ be a function such that $V(x)<\infty$ for $\Pi$-a.s.
The transition kernel $P$ is said to have the \textbf{geometric drift condition} if there is a small set $C$, $\gamma\in (0,1)$ and $b<\infty$
\begin{align}\label{eq:drift_condition}
	PV\le \gamma V+b1_C.
\end{align}
 The condition is extensively studied in the past few decades. In particular, if the above condition is satisfied, and also there exists a small set  that satisfies (\ref{eq:smallset}) for $k=1$, then $P$ is \textbf{geometrically ergodic}, that is 
\begin{align*}
	\|P^n-\Pi\|_{V}\le c\gamma^nV(x)
\end{align*}
where $\|\nu\|_{V}=\sup_{f:|f|\le V}|\int f(x)\nu(\dif x)|$ for a signed measure $\nu$
and $c$ is a constant (See Theorem 15.0.1 of \citet{MR2509253}). Moreover, geometric ergodicity implies geometric drift condition if the conditions (a) and (c) are satisfied (See Theorem 16.0.1 of \citet{MR2509253}).

\subsection{Necessary condition for ergodicity}\label{sec:necessary}
In this section we introduce necessary condition for geometric ergodicity for random-walk type kernels (RWM and MpCN) and MH kernel with ergodic proposal kernel (pCN).  
Let $(E,d)$ be a pseudometric space, that is, $d(x,y)\ge 0, d(x,y)=d(y,x)$
and $d(x,z)\le d(x,y)+d(y,z)$. Let $B_r(x)=\{y;d(x,y)<r\}$. Let $\mathcal{E}$ be its Borel $\sigma$-algebra generated by the pseudometric topology. 
Fix $x^*\in E$. 
The RWM kernel, and the MpCN kernel after transform $\Psi$ satisfy the following property. 

\begin{assumption}\label{ass:rwm}
	For any $\epsilon>0$, there exists $r>0$ such that 
	$P(x, B_r(x))>1-\epsilon$ for any $x\in E$. 
\end{assumption}

The following proposition,  due to \cite{MR1996270}, gives necessary condition for ergodicity. This says that if $P$ is geometrically ergodic, the target distribution has exponential tail. We give a proof for the sake of convenience of the reader.

\begin{proposition}[\cite{MR1996270}]\label{pro:integrable}
Assume Assumption \ref{ass:rwm}. 
	If $P$ satisfies geometric ergodicity then there exists $\delta>0$ such that
	\begin{align*}
		\int_{x\in E} \exp(\delta d(x^*,x))\Pi(\dif x)<\infty.
	\end{align*}
\end{proposition}

To prove the proposition, we need two simple lemmas. The first lemma says that  small set is ``small''. The second lemma says that under geometric ergodicity, there is a small set ``large enough''. 

\begin{lemma}\label{lem:smallset}
	Under Assumption \ref{ass:rwm}, any small set is bounded. 
\end{lemma}

\begin{proof}
Assume by contradiction that there is an unbounded small set $C$ such that (\ref{eq:smallset}). Choose $r>0$ so that $P(x, B_r(x))>(1-\delta/2)^{1/k}$. Then by definition, by putting $x_0=x$, 
\begin{align*}
	P^k(x,B_{kr}(x))\ge \int \prod_{l=1}^kI_{B_r(x)}(x_l)P(x_{l-1},\dif x_l)> \overbrace{\left(1-\delta/2\right)^{1/k}\cdots \left(1-\delta/2\right)^{1/k}}^k=1-\delta/2
\end{align*}
and thus $P^k(x, B_{kr}(x)^c)\le\delta/2$. 
Since $C$ is unbounded, we can choose $x_i\in C\ (i=1,2)$ so that $B_{kr}(x_1)\cap B_{kr}(x_2)=\emptyset$. 
Then 
\begin{align*}
	P^k(x_i,B_{kr}(x_i)^c)\ge \delta\nu(B_{kr}(x_i)^c)
\end{align*}
and hence $\nu(B_{kr}(x_i))=1-\nu(B_{kr}(x_i)^c)\ge 1-\delta^{-1}P^k(x_i,B_{kr}(x_i)^c)> 1/2$ for $i=1,2$. This would imply $\nu(E)\ge \nu(B_{kr}(x_1))+\nu(B_{kr}(x_2))>1$ which is a contradiction. Thus any small set is bounded. 
\end{proof}

\begin{lemma}\label{lem:integrable} 
	If transition kernel $P$ satisfies (\ref{eq:drift_condition}), then there exists $s>1$ such that
	\begin{align*}
		\sum_{n=0}^\infty s^n(I_{C^c}P)^n(x,E)
	\end{align*}
	is $\Pi$-integrable. 
\end{lemma}

\begin{proof}
By  (\ref{eq:drift_condition}), 
	\begin{align*}
		(I_{C^c}P)V=I_{C^c}(PV)\le I_{C^c}(\gamma V+b1_C)=I_{C^c}\gamma V\le \gamma V. 
	\end{align*}
	Then by choosing $s>1$ so that $s\gamma<1$, we have
	\begin{align*}
		\sum_{n=0}^\infty s^n(I_{C^c}P)^n(x,E)\le
		\sum_{n=0}^\infty s^n(I_{C^c}P)^nV\le \sum_{n=0}^\infty s^n\gamma^nV\le\frac{1}{1-s\gamma}V. 
	\end{align*}
	By assumption, $P$ is irreducible and $\Pi$-invariant. 
By Theorem 14.3.7 of \cite{MR2509253}, the drift function $V$ in (\ref{eq:drift_condition}) is $\Pi$-integrable.
	Hence the left-hand side is also $\Pi$-integrable 
\end{proof}

\begin{proof}[Proof of Proposition \ref{pro:integrable}]
By Lemma \ref{lem:smallset}, small set $C$ is bounded. 
We choose $r>0$ such a way that $C\subset B_r(x^*)$
and $P(x, B_r(x))>s^{-1}\ (x\in E)$ where $s>1$ is as in Lemma \ref{lem:integrable}. 
	If $x_0\notin B_{nr}(x^*)$ and if $x_n\in B_r(x_{n-1})\ (n\ge 1)$, then 
	$x_0,\ldots, x_{n-1}\notin B_r(x^*)$. Therefore
	\begin{align*}
		(I_{C^c}P)^n(x_0, E)\ge (I_{B_r(x^*)^c}P)^n(x_0, E)\ge \int \prod_{m=1}^{n}I_{B_r(x_{m-1})}(x_m)P(x_{m-1},\dif x_m)\ge (\inf_xP(x,B_r(x)))^n=:\eta^n
	\end{align*}
	where $\eta>s^{-1}$. Thus 
	\begin{align*}
		s^{n(x)}(I_{C^c}P)^{n(x)}(x,E)\ge  (s\eta)^{n(x)}=\exp(n(x)\log(s\eta))
	\end{align*}
	 where $n(x)=\left[d(x^*,x)/r\right]\ge d(x^*,x)/r-1$
	where $[t]$ is the integer part of $t>0$. 
	Therefore we can find $c>0$ and $\delta>0$ such that 
	\begin{align*}
		s^{n(x)}(I_{C^c}P)^{n(x)}(x,E)\ge  c\exp(\delta d(x^*,x)). 
	\end{align*}
	Since the left-hand side is $\Pi$-integrable by Lemma \ref{lem:integrable}, the right-hand side is also $\Pi$-integrable. 
\end{proof}

By Proposition \ref{pro:integrable}, RWM kernel is geometrically ergodic only if $\Pi$ has a light-tailed density. The MpCN kernel has the same property but after the projection $x\mapsto \log\|x\|^2$. The requirement of the MpCN kernel is that $\Pi$ has a polynomial-tailed density  which is much weaker condition compared to the RWM kernel. 

\begin{corollary}[\cite{MR1996270}]\label{cor:rwm}
The random-walk Metropolis kernel on $\mathbb{R}^d$ satisfy Assumption \ref{ass:rwm} for Euclidean metric $d(x,y)$. Thus
if the kernel is geometrically ergodic, by taking $x^*=0$, 
	\begin{align}\label{eq:mgf}
		\int \exp(\delta\|x\|)\Pi(\dif x)<\infty
	\end{align}
	for some $\delta>0$, where $\|\cdot\|$ is Euclidean norm. 
\end{corollary}

\begin{corollary}\label{cor:mpcn}
The MpCN kernel on $\mathbb{R}^d$ satisfy Assumption \ref{ass:rwm} for $d(x,y)=|\log(\|x\|^2)-\log(\|y\|^2)|$ on where $\|\cdot\|$ is Euclidean norm. Thus
if the kernel is geometrically ergodic, by taking $x^*\in S^{d-1}$, 
	\begin{align*}
		\int \exp\left(\delta\left|\log\left(\|x\|^2\right)\right|\right)\Pi(\dif x)<\infty
	\end{align*}
	for some $\delta>0$. In particular, 
		\begin{align*}
		\int \|x\|^\delta\Pi(\dif x)<\infty. 
	\end{align*}
\end{corollary}

Proposition \ref{pro:integrable} is useful for MH kernels with transient proposal kernels, but may not be useful for those with ergodic proposal kernels. In order to study necessary condition for the latter case, we need to estimate  the acceptance probability. There is a useful result due to \citet{MT2,RT}. 

\begin{proposition}[\cite{RT}]\label{pro:pcn}
If $P$ is geometrically ergodic, then $\Pi$-$\mathrm{ess}\sup P(x,\{x\})<1$. 
\end{proposition}

\begin{proof}
Let $E'=\{x; V(x)<\infty\}$. 
To obtain a contradiction, suppose $\Pi$-$\mathrm{ess}\sup P(x,\{x\})=1$.
Choose $x_n\in E'$ so that $P(x_n,\{x_n\})\ge 1-n^{-1}$. Any small set $C$ only includes finitely many elements of $\{x_n\}_n$. Otherwise, if (\ref{eq:smallset}) is satisfied, then 
	\begin{align*}
		\delta\nu\left(\{x_n;n\ge N\}^c\right)\le 
		\delta\nu\left(\{x_N\}^c\right)\le 
		P^k(x_N,\{x_N\}^c)\le 1-\left(1-\frac{1}{N}\right)^k
	\end{align*}
	for each $x_N\in C$. Taking $N\rightarrow\infty$ we have $\delta=0$ and hence this contradicts $C$ is a small set. 
	
	By geometric ergodicity, (\ref{eq:drift_condition}) is satisfied. Choose $x_n$ as above such that $x_n\notin C$. Then 
	\begin{align*}
		\left(1-\frac{1}{n}\right)V(x_n)\le P(x_n,\{x_n\})V(x_n)\le PV(x_n)\le \gamma V(x_n). 
	\end{align*}
	By taking $n\rightarrow\infty$, $\gamma=1$ and hence this contradicts our assumption for geometric ergodicity of $P$. 
\end{proof}

We state a necessary  condition for ergodicity for the pCN kernel as a corollary of Proposition \ref{pro:pcn}. It says that the pCN kernel requires even lighter-tailed density  for the target distribution than the RWM kernel.

\begin{corollary}\label{cor:pcn}
Suppose that $\Pi$ has a probability density $\pi(x)$. 
	For each $r>0$, let  
	\begin{align*}
		C_r=r^{-2}\sup_{\rho r\le \|x\|,\|y\|\le r}|\log\pi(x)-\log\pi(y)|. 
	\end{align*}
	If the pCN kernel is geometrically ergodic, then $\liminf_{r\rightarrow\infty}C_r\ge(1-\rho)/2$. 
\end{corollary}

\begin{proof}
Write $x^*=\rho^{1/2}x+(1-\rho)^{1/2}w$. 
Assume that the pCN kernel is geometrically ergodic. By Proposition \ref{pro:pcn}, 
\begin{align}\label{eq:pcnacceptprob}
	\delta<1-P(x,\{x\})=\int \min\left\{1,\frac{\pi(x^*)\phi_d(x)}{\pi(x)\phi_d(x^*)}\right\}\phi_d(w)\dif w
\end{align}
for (Leb) a.s. $x$ for some $\delta>0$. 
By triangular inequality, for sufficiently large $r=\|x\|$, 
we have $\rho r\le \|x^*\|\le r$ since
\begin{align*}
	\|x^*\|=\|\rho^{1/2}x+(1-\rho)^{1/2}w\|=\rho^{1/2}\|x\|+o(\|x\|)
\end{align*}
with an obvious inequality $\rho<\rho^{1/2}<1$. 
 Thus for each $w\in \mathbb{R}^d$, 
\begin{align*}
\log\left(\frac{\pi(x^*)\phi_d(x)}{\pi(x)\phi_d(x^*)}\right)&=
\left\{\log\pi(x^*)-\log\pi(x)\right\}+
\left\{\log\phi_d(x)-\log\phi_d(x^*)\right\}\\
&=
\left\{\log\pi(x^*)-\log\pi(x)\right\}+
\left\{-\frac{\|x\|^2}{2}+\frac{\rho\|x\|^2+2\sqrt{\rho(1-\rho)}\langle x,w\rangle+(1-\rho)\|w\|^2}{2}\right\}\\
&\le C_rr^2-\frac{1-\rho}{2}r^2+O(r)\\
&= r^2\left\{C_r-\frac{1-\rho}{2}+O(r^{-1})\right\}. 
\end{align*}
Therefore if $\liminf C_r<(1-\rho)/2$, we can choose a sequence of $r=r_n=\|x_n\|$ such that the right-hand side of the above tends to $0$. 
By Lebesgue's dominated convergence theorem, the right-hand side of (\ref{eq:pcnacceptprob}) converges to $0$  for this sequence, which is a contradiction. 
Thus $\liminf C_r\ge (1-\rho)/2$. 
\end{proof}

\subsection{Ergodicity for regular varying function}\label{sec:regular}

We prove geometric ergodicity in terms of regularly varying property. For introductory literature to regularly varying functions we refer the reader to the books \citet{MR1015093,MR2364939}. 
The theory of regularly varying functions provides a framework for heavy-tail analysis. 
For one dimensional case, a positive function $h(r)$ on $(0,\infty)$ is called 
regularly varying if $\lim_{r\rightarrow\infty}h(rx)/h(r)=\lambda(x)$
for some positive finite valued function $\lambda$. We consider multidimensional  version. 
We denote $a(r,x)\xrightarrow{ucp}a(x)\ (r\rightarrow\infty)$ if for any $x\in \mathbb{R}^d\backslash\{0\}$ there exists a compact set $K\ni x$ such that $\lim_{r\rightarrow\infty}\sup_{y\in K}|a(r,y)-a(y)|=0$. 

\begin{definition}
Positive valued function $h(x)\ (x\in\mathbb{R}^d)$ is called symmetrically regularly varying if
\begin{align*}
\frac{h(rx)}{h(r1)} \xrightarrow{ucp}\lambda(x)\ (r\rightarrow\infty)
\end{align*}
for some $\lambda:\mathbb{R}\rightarrow (0,\infty)$ such that $\lambda(x)=1$ for any $x\in S^{d-1}$ where $1=(1,\ldots, 1)\in\mathbb{R}^d$. 
\end{definition}

This class includes many functions such as polynomial target densities considered in \citet{MR2396939} (Section 3.3). 
This class inherits useful properties  from one dimensional regularly varying function: 
 $\lambda(x)=\|x\|^{-\alpha}$ for the \textbf{exponent of variation}   $-\alpha\in\mathbb{R}$ (p277 of \citet{MR2364939}). 
 Note that symmetricity of $\lambda(x)$ is crucial in our proof. It is not obvious to construct a  simple sufficient condition for geometric ergodicity for non-symmetric case. 

Assume that  $\Pi$ has the density $\pi(x)$. Before stating the main result of this section we prove simple lemma for integrability of the regularly varying function.

\begin{lemma}\label{lem:regular}
	If $\int \|x\|^\delta \Pi(\dif x)<\infty$ for some $\delta>0$, then 
	the exponent of variation  $-\alpha$ of the symmetrically regularly varying function $\pi(x)$ satisfies
	$\alpha>d$. 
\end{lemma}

\begin{proof}
	For $x>0$, let $h(r)=\int_{\xi\in S^{d-1}}\pi(r\xi)r^{d-1}\dif\xi$. 
	Then $h(r)$ is a regularly varying function with exponent of variation $-\alpha+d-1$ by local uniform convergence property since
	\begin{align*}
		\frac{h(rs)}{h(r)}=\frac{h(rs)/\pi(r1)}{h(r)/\pi(r1)}=\frac{\int_{\xi\in S^{d-1}}\pi(rs\xi)/\pi(r1)(rs)^{d-1}\dif\xi}{\int_{\xi\in S^{d-1}}\pi(r\xi)/\pi(r1)r^{d-1}\dif\xi}\rightarrow \frac{\int_{\xi\in S^{d-1}}\lambda(s\xi)(rs)^{d-1}\dif \xi}{\int_{\xi\in S^{d-1}}\lambda(\xi)r^{d-1}\dif \xi}=s^{-\alpha+d-1}\ (r\rightarrow\infty). 
	\end{align*}
	Therefore, by Potter bounds (Theorem 1.5.6 (iii) of \citet{MR1015093}), if $\beta-\alpha+d-1>-1$, then $\int_1^\infty r^\beta h(r)\dif r=\infty$.
	Since $\int_1^\infty r^\delta h(r)\dif r=\int_{\|x\|>1} \|x\|^\delta\Pi(\dif x)<\infty$, 
	we have $\delta-\alpha+d-1\le -1$. Thus $\delta+d\le \alpha$, and hence $d<\alpha$. 
\end{proof}

\begin{proposition}
If $\pi(x)$ is symmetrically regularly varying function, then RWM kernel and pCN kernel do not have geometric ergodicity. 	
\end{proposition}

\begin{proof}œ
Let $h(r)$ be as in the previous  lemma. 
	Then $h(r)$ is a regularly varying function and hence $\int e^{s\|x\|}\Pi(\dif x)=\int e^{sr}h(r)\dif r=+\infty$ (Theorem 1.5.6 (iii) of \citet{MR1015093}). 
	Hence RWM kernel does not have geometric ergodicity by Corollary \ref{cor:rwm}. 
	
Next we consider pCN kernel. By local uniform convergence property, 
\begin{align*}
	r^2C_r&=\sup_{\rho\le\|x\|,\|y\|\le 1}|\log\pi(rx)-\log\pi(ry)|\\
	&=\sup_{\rho\le\|x\|,\|y\|\le 1}\left|\left(\log\frac{\pi(rx)}{\pi(r1)}-\log\lambda(x)\right)-\left(\log\frac{\pi(ry)}{\pi(r1)}-\log\lambda(y)\right)+\left(\log\lambda(x)-\log\lambda(y)\right)\right|\\
	&\le2\sup_{\rho\le\|x\|\le 1}\left|\log\frac{\pi(rx)}{\pi(r1)}-\log\lambda(x)\right|
	+\sup_{\rho\le\|x\|,\|y\|\le 1}\left|\log\lambda(x)-\log\lambda(y)\right|\\
	&\rightarrow\sup_{\rho\le\|x\|,\|y\|\le 1}\left|\log\lambda(x)-\log\lambda(y)\right|<\infty\ (t\rightarrow\infty). 
\end{align*}
Thus $C_r=o(1)$, and hence pCN kernel does not have geometric ergodicity by Corollary \ref{cor:pcn}. 
\end{proof}

By Corollary \ref{cor:mpcn}, the MpCN kernel is geometrically ergodic only if $\Pi$ has a polynomial tail. The following proposition states the converse. 

\begin{proposition}\label{pro:ergreg}
Assume $\pi(x)$ is strictly positive continuous symmetrically regularly varying function. 
Then the MpCN kernel is geometrically ergodic if and only if 
$\int \|x\|^\delta \Pi(\dif x)<\infty$ for some $\delta>0$. 
\end{proposition}

\begin{proof}
We use expression in (\ref{eq:xstar}) and (\ref{eq:xi}). 
Let $q(x)=\pi(x)\|x\|^d$ and 
let 
$V(x)=
\left\{
\begin{array}{cl}
q(x)^{-s}&x\neq 0\\
+\infty&x=0	
\end{array}
\right.
$ 
for $s\in (0,1)$. 
Then $V(x)$ is bounded on $C=\{x;r\le \|x\|\le r^{-1}\}$ for any $r\in (0,1)$, 
and $C$ is a small set for the MpCN kernel. 
To prove (\ref{eq:drift_condition}), 
it is sufficient to show
\begin{align}\label{eq:drift}
	\limsup_{x\rightarrow 0} \frac{PV(x)-V(x)}{V(x)}<0,\ 
		\limsup_{x\rightarrow \infty} \frac{PV(x)-V(x)}{V(x)}<0. 
\end{align}
Observe that
\begin{align}\label{eq:ergregdrift}
	\frac{PV(x)-V(x)}{V(x)}=\mathbb{E}_x\left[\left\{\left(\frac{q(x^*)}{q(x)}\right)^{-s}-1\right\}\min\left\{1,\frac{q(x^*)}{q(x)}\right\}\right]
\end{align}	
and the integrand is uniformly bounded. Since $\pi(x)$ is continuous at $0$, for each $w, \tilde{w}$, 
\begin{align*}
1=\lim_{x\rightarrow 0} \frac{\pi(x^*)}{\pi(x)}=\lim_{x\rightarrow 0} \frac{q(x^*)}{q(x)}\exp\left(-\frac{d}{2}\xi(x)\right). 
\end{align*}
 Since the law of $\xi(x)$ is independent of $x$, we simply write $\xi$ for $\xi(x)$. 
 Then by Slutsky's theorem, $q(x^*)/q(x)$ converges in law to $\exp(d\xi/2)$ as $x\rightarrow 0$. 
Therefore we have
\begin{align*}
	\lim_{x \rightarrow 0}\frac{PV(x)-V(x)}{V(x)}&=\mathbb{E}\left[\left\{e^{-\frac{ds}{2}\xi}-1\right\}\min\left\{1,e^{\frac{d}{2}\xi}\right\}\right]\\
	&=\mathbb{E}\left[\left\{e^{-\frac{ds}{2}\xi}-1\right\},\xi>0\right]
	+\mathbb{E}\left[e^{\frac{d(1-s)}{2}\xi}\left\{1-e^{\frac{ds}{2}\xi}\right\},\xi<0\right]. 
	\end{align*}
	By Proposition \ref{prop:rwp}, the law of $\xi$ is symmetric about the origin. Therefore the above expectation equals to 
	\begin{align*}
		\mathbb{E}\left[\left\{e^{-\frac{ds}{2}\xi}-1\right\},\xi>0\right]
	+\mathbb{E}\left[e^{-\frac{d(1-s)}{2}\xi}\left\{1-e^{-\frac{ds}{2}\xi}\right\},\xi>0\right]
			&=\mathbb{E}\left[\left(1-e^{-\frac{d(1-s)}{2}\xi}\right)\left\{e^{-\frac{ds}{2}\xi}-1\right\},\xi>0\right]<0 
\end{align*}	
since $\mathbb{P}(\xi>0)>0$ and the integrand is negative for any $\xi>0$. Thus the first part of (\ref{eq:drift}) is completed. 

Now we consider the second part of (\ref{eq:drift}). Let $\eta(x)=(d-\alpha)\xi(x)/2$. Then for each $w, \tilde{w}$, 
 \begin{align*}
\frac{q(x^*)}{q(x)}\exp\left(-\eta(x)\right)
 	 	=\frac{\pi(x^*)}{\pi(x)}\exp\left(\frac{\alpha}{2}\xi(x)\right)
 	 	=\frac{\pi\left(\|x\|\frac{x^*}{\|x\|}\right)/\pi(\|x\|\cdot 1)}{\pi\left(\|x\|\frac{x}{\|x\|}\right)/\pi(\|x\|\cdot 1)}\exp\left(\frac{\alpha}{2}\xi(x)\right).
 \end{align*}
 By local uniform convergence of the regular varying function, the right-hand side of the above converges to $1$. Since the law of $\eta(x)$ does not depend on $x$, we simply denote it by $\eta$. Thus
 as in the first part of (\ref{eq:drift}), $q(x^*)/q(x)$ converges in law to $\exp(\eta)$ by Slutsky's  theorem as $x\rightarrow\infty$. Hence 
\begin{align*}
	\lim_{x \rightarrow \infty}\frac{PV(x)-V(x)}{V(x)}&=\mathbb{E}\left[\left\{e^{-s\eta}-1\right\}\min\left\{1,e^{\eta}\right\}\right]\\
	&=\mathbb{E}\left[\left\{e^{-s\eta}-1\right\},\eta>0\right]
	+\mathbb{E}\left[e^{(1-s)\eta}\left\{1-e^{s\eta}\right\},\eta< 0\right]. 
 \end{align*}
Since the distribution of $\eta$ is symmetric, the above expectation equals to 
\begin{align*}
		\mathbb{E}\left[\left\{e^{-s\eta}-1\right\},\eta>0\right]
	+\mathbb{E}\left[e^{-(1-s)\eta}\left\{1-e^{-s\eta}\right\},\eta> 0\right]=\mathbb{E}\left[\left(1-e^{-(1-s)\eta}\right)\left\{e^{-s\eta}-1\right\},\eta>0\right]. 
\end{align*}	
Since the integrand is negative 
if $\eta>0$, the claim follows if $\mathbb{P}(\eta>0)>0$. 
Since $\mathbb{P}(\eta\neq 0)=2\mathbb{P}(\eta>0)$, 
we have geometric ergodicity if $\mathbb{P}(\eta=0)<1$. 
However $\mathbb{P}(\eta=0)=1$ is satisfied if and only if $\alpha=d$, which contradicts the assumption by Lemma \ref{lem:regular}. Thus the claim follows. 
\end{proof}

\subsection{Ergodicity for rapidly varying function}\label{sec:rapid}

In this section we illustrate ergodic property for the MpCN kernel for light-tailed target distributions.  We show that the MpCN kernel is geometric ergodicity for any light-tailed target distribution as long as the curvature condition (\ref{eq:curvature}) is satisfied. On the other hand, as in Corollaries \ref{cor:rwm} and \ref{cor:pcn},  super-exponential tail is necessary for the RWM kernel and the pCN kernel. 
To sate the main result, we need a definition for light-tailed distributions. 

\begin{definition}
A positive function $h(x)\ (x\in\mathbb{R}^d)$ is rapidly varying if 
for any $\xi\in S^{d-1}$, $s>0$, 
\begin{align*}
	\lim_{r\rightarrow\infty}\frac{h(rs\xi)}{h(r\xi)}=
	\left\{\begin{array}{cc}
	0&\mathrm{if}\ 1<s\\
	+\infty&\mathrm{if}\ 1>s
	\end{array}
\right.
\end{align*}
\end{definition}

Many light-tailed functions are rapidly varying. 
For example, super-exponential densities in \citet{MR1731030} (Section 4) and exponential densities in \citet{FM2} (Assumption D) are rapidly varying functions. 
See also \citet{MR3097969} for other examples.  

\begin{proposition}\label{prop:mpcn_suffice}
If $\pi(x)$ is a  continuous strictly positive rapidly varying function, MpCN kernel is geometrically ergodic
if and only if $\mathop{\mbox{$\Pi$-$\mathrm{ess}\sup$}} P(x,\{x\})<1$. 
\end{proposition}


\begin{proof}
The only if part follows from Proposition \ref{pro:pcn}. 
The proof wil be finished once we show (\ref{eq:rapid}) in Proposition \ref{pro:ergreg} if $\mathop{\mbox{$\Pi$-$\mathrm{ess}\sup$}} P(x,\{x\})<1$. We only show the latter inequality in (\ref{eq:rapid}) since the proof for the former inequality is exactly the same as that of Proposition \ref{pro:ergreg}.
Thanks to $\lim_{x\rightarrow \pm\infty}e^{-sx}\min\{1,e^x\}=0$ if 
\begin{align}\label{eq:rapid}
\lim_{x\rightarrow\infty}\mathbb{P}_x\left[\left|\log \frac{q(x^*)}{q(x)}\right|\le C\right]	=0
\end{align}
for any $C>0$, then 
\begin{align*}
\lim_{x\rightarrow\infty}\mathbb{E}_x\left[\left\{\frac{q(x^*)}{q(x)}\right\}^{-s}\min\left\{1,\frac{q(x^*)}{q(x)}\right\}\right]=0. 
\end{align*}
Therefore by the expression (\ref{eq:ergregdrift}), the equation (\ref{eq:rapid}) implies
\begin{align*}
\limsup_{x\rightarrow\infty}\frac{PV(x)-V(x)}{V(x)}=-\liminf_{x\rightarrow\infty}\mathbb{E}_x\left[\min\left\{1,\frac{q(x^*)}{q(x)}\right\}\right]
=\limsup_{x\rightarrow\infty}P(x,\{x\})-1=\mathop{\mbox{$\Pi$-$\mathrm{ess}\sup$}} P(x,\{x\})-1
\end{align*}
where the last equality comes from continuity of $\pi(x)$. 
Thus (\ref{eq:rapid}) will complete the proof. 

By Proposition \ref{prop:rwp}, the law of $\|x^*\|/\|x\|=\exp(\xi/2)$ does not depend on $x$. Therefore there exists $\delta>0$
for each $\epsilon>0$ such that
\begin{align*}
\mathbb{P}_x\left[\left|\log \frac{q(x^*)}{q(x)}\right|\le C\right]&\le 
\mathbb{P}_x\left[\left|\log \frac{q(x^*)}{q(x)}\right|\le C,\delta\le \frac{\|x^*\|}{\|x\|}\le \delta^{-1}\right]
+
\mathbb{P}_x\left[\left\{\delta\le \frac{\|x^*\|}{\|x\|}\le \delta^{-1}\right\}^c\right]\\
&\le
\mathbb{P}_x\left[\frac{x^*}{\|x\|}\in A(x)\right]
+
\epsilon
\end{align*}
where 
\begin{align*}
A(x)=\left\{y\in\mathbb{R}^d;\left|\log \frac{q(\|x\|y)}{q(x)}\right|\le C,\delta\le \|y\|\le \delta^{-1}\right\}. 
\end{align*}
By expression (\ref{eq:xstar}), $x^*/\|x\|$ follows the multivariate Cauchy distribution
with shift $\rho^{1/2}x/\|x\|$ and scale $1-\rho$. Thus the probability distribution function is uniformly bounded, and hence there exists $c>0$ such that
	\begin{align}\label{eq:xstarparx}
	\mathbb{P}_x\left[\frac{x^*}{\|x\|}\in A(x)\right]\le c\mathrm{Leb}(A(x))
	= c\int_{\xi\in S^{d-1}}\int_{r\in A(x,\xi)}r^{d-1}\dif r\dif\xi
		\end{align}
		where 
\begin{align*}
A(x,\xi)=\left\{r>0;\left|\log \frac{q(\|x\|r\xi)}{q(x)}\right|\le C,\delta\le r\le \delta^{-1}\right\}. 
\end{align*}
By dominated convergence theorem, (\ref{eq:xstarparx}) tends to $0$ if $\lim_{x\rightarrow\infty}\mathrm{Leb}(A(x,\xi))=0$
for each $\xi\in S^{d-1}$. 
Note that
\begin{align*}
A(x,\xi)\times A(x,\xi)&= \left\{r,s>0;\left|\log \frac{q(\|x\|r\xi)}{q(x)}\right|, \left|\log \frac{q(\|x\|s\xi)}{q(x)}\right|\le C,\delta\le r,s\le \delta^{-1}\right\}\\
&\subset
\left\{r,s>0;\left|\log \frac{q(\|x\|r\xi)}{q(\|x\|s\xi)}\right|\le 2C,\delta\le r,s\le \delta^{-1}\right\}. 
\end{align*}
However by rapidly  varying property of $\pi(x)$, 
\begin{align*}
\liminf_{x\rightarrow\infty}\left|\log \frac{q(\|x\|r\xi)}{q(\|x\|s\xi)}\right|
\ge \liminf_{x\rightarrow\infty}\left|\log \frac{\pi(\|x\|r\xi)}{\pi(\|x\|s\xi)}\right|
-\left|\log \frac{(\|x\|r)^d}{(\|x\|s)^d}\right|\rightarrow\infty	
\end{align*}
for each $\delta\le r,s,\le\delta^{-1}$ such that $r\neq s$, and hence 
$\mathrm{Leb}(A(x,\xi))=\sqrt{\mathrm{Leb}(A(x,\xi)\times A(x,\xi))}\rightarrow 0$. 
Thus the claim follows. 
\end{proof}

We state the main result in this section. 
The  curvature condition considered in \cite{MR1731030}  is the sufficient condition for geometric ergodicity for the MpCN kernel. The proof follows a similar line of argument to \cite{MR1731030}, proof of Theorem 4.3. 

\begin{corollary}
	If $\pi(x)$ is  differentiable and strictly positive rapidly varying function and if
	\begin{align}\label{eq:curvature}
		\limsup_{x\rightarrow\infty}\left\langle\frac{x}{\|x\|},\frac{\nabla\log\pi(x)}{\|\nabla\log\pi(x)\|}\right\rangle <0
	\end{align}
	then MpCN kernel is geometrically ergodic. 
\end{corollary}

\begin{proof}
	Let $n(x)=x/\|x\|$ and $m(x)=\nabla\log \pi(x)/\|\nabla\log \pi(x)\|$. 
	By assumption, there exists $\epsilon\in (0,1)$ and $M>0$ such that
	\begin{align*}
		\left\langle n(x),m(x)\right\rangle< -2\epsilon	
	\end{align*}
	for all $\|x\|\ge M$. 
	Let
	\begin{align*}
		W(x)=\{x-t\|x\|\xi;\ 0\le t\le \epsilon^2/4, \xi\in S^{d-1}, \|\xi-n(x)\|\le\epsilon\}. 
	\end{align*}
	We first prove that 
	\begin{align}\label{eq:curvature_w}
		\left(y\in W(x)\ \mathrm{and}\ \|x\|\ge 2M \right) \Rightarrow \log \frac{\pi(y)}{\pi(x)}>0.
	\end{align}
	If $y=x-t\|x\|\xi\in W(x)$ 
	then
	\begin{align*}
		\|n(x)-n(y)\|^2&=2-2\langle n(x),n(y)\rangle=2-2\left\langle n(x),\frac{n(x)-t\xi}{\|n(x)-t\xi\|}\right\rangle\\
		&\le 2-2\frac{1-t}{1+t}=\frac{4t}{1+t}\le 4t\le \epsilon^2
	\end{align*}
	Since $1-t\ge 1-\epsilon^2/4\ge 1/2$, for $\|x\|\ge 2M$ we have $\|y\|=\left\|x-t\left\|x\right\|\xi\right\|\ge (1-t)\|x\|\ge M$. Then for $y\in W(x)$ and $\|x\|\ge 2M$, 
	\begin{align*}
	\langle\xi,m(y)\rangle 	&=\langle (\xi-n(x))+(n(x)-n(y))+n(y),m(y)\rangle\\
	&\le \|\xi-n(x)\|+\|n(x)-n(y)\|+\langle n(y),m(y)\rangle<0. 
	\end{align*}
	Hence (\ref{eq:curvature_w}) holds, since if $y\in W(x)$ and $\|x\|\ge 2M$, then 
	\begin{align*}
	\log\frac{\pi(y)}{\pi(x)}=\log \pi(y)-\log \pi(x)=-\|x\|\int_0^t\left\langle \xi,\nabla\log \pi(x-s\|x\|\xi)\right\rangle\dif s>0. 
	\end{align*}
		Next we prove
		\begin{align}\label{eq:curvature_q}
	\liminf_{x\rightarrow\infty} \mathbb{P}_x\left(\log\frac{\pi(x^*)}{\pi(x)}>0\right)\le 
		\liminf_{x\rightarrow\infty} \mathbb{P}_x\left(\log\frac{q(x^*)}{q(x)}>0\right). 
	\end{align}
	By the RWM property of the MpCN kernel (Proposition \ref{prop:rwp}), the law of $\xi(x)/2=\log(\|x^*\|/\|x\|)$ does not depend on $x$. Therefore
	for any $\delta>0$ there exists $C>0$ such that
	\begin{align*}
	\mathbb{P}_x\left(d\log\frac{\|x^*\|}{\|x\|}\le -C\right)	<\delta. 
	\end{align*}
 	Therefore by (\ref{eq:rapid})
	\begin{align*}
	\liminf_{x\rightarrow\infty} \mathbb{P}_x\left(\log\frac{\pi(x^*)}{\pi(x)}>0\right)&=
	\liminf_{x\rightarrow\infty} \mathbb{P}_x\left(\log\frac{q(x^*)}{q(x)}>d\log\frac{\|x^*\|}{\|x\|}\right)\\
	&\le
		\liminf_{x\rightarrow\infty} \mathbb{P}_x\left(\log\frac{q(x^*)}{q(x)}>-C\right)+\mathbb{P}_x\left(d\log\frac{\|x^*\|}{\|x\|}\le -C\right)\\
	&\le
		\liminf_{x\rightarrow\infty} \mathbb{P}_x\left(\log\frac{q(x^*)}{q(x)}>0\right)+\delta.		
	\end{align*}
	Hence (\ref{eq:curvature_q}) is proved. 
	Thus for $\|x\|\ge 2M$, by (\ref{eq:curvature_w}) and (\ref{eq:curvature_q}), 
		\begin{align*}
P(x,\{x\}^c)&=\mathbb{E}_x\left[\min\left\{1,\frac{q(x^*)}{q(x)}\right\}\right]\\
&\ge \mathbb{P}_x\left(\log\frac{q(x^*)}{q(x)}>0\right)\\
&\ge \mathbb{P}_x\left(\log\frac{\pi(x^*)}{\pi(x)}>0\right)+o(\|x\|)\\
&\ge \mathbb{P}_x\left(x^*\in W(x)\right)+o(\|x\|). 
	\end{align*}
Observe that $UW(x)=\{Uy; y\in W(x)\}=W(Ux)$ for any unitary matrix $U$, and $x\in W(x)\Leftrightarrow e\in W(e)$ for $e=x/\|x\|$. By these facts
		\begin{align*}
\mathbb{P}_x\left(x^*\in W(x)\right)=
\mathbb{P}_e\left(x^*\in W(e)\right)>0
	\end{align*}
for any $e\in S^{d-1}$. 
Thus $\liminf_{x\rightarrow\infty}P(x,\{x\}^c)=\mathop{\mbox{$\Pi$-$\mathrm{ess}\sup$}} P(x,\{x\})-1\ge \mathbb{P}_e\left(x^*\in W(e)\right)>0$. Thus the claim follows by Proposition \ref{prop:mpcn_suffice}. 
\end{proof}

%


\section*{Acknowledgement}
The author would like to extend thanks to  Alexandros Beskos for interesting discussions related to the similarity of MpCN and MMALA algorithms. 
I also thank to  Masayuki Uchida for fruitful discussions related to practical implementation of MpCN for inference of discretely observed stochastic diffusion process. 

\bibliographystyle{plainnat}

\end{document}